\documentclass{article}

\usepackage{arxiv}

\usepackage[utf8]{inputenc} 
\usepackage[T1]{fontenc}    
\usepackage{hyperref}       
\usepackage{url}            
\usepackage{booktabs}       
\usepackage{amsfonts}       
\usepackage{nicefrac}       
\usepackage{microtype}      
\usepackage{lipsum}

\usepackage{cite}
\usepackage{amsmath,amssymb,amsfonts}
\usepackage{graphicx}
\usepackage{algorithm}
\usepackage[noend]{algpseudocode}
\usepackage{hyperref}
\usepackage{textcomp}
\usepackage{threeparttablex} 

\usepackage{mathtools}
\usepackage{caption}
\usepackage{float}
\usepackage{stmaryrd}
\usepackage{epstopdf}
\usepackage{multirow}
\usepackage{pifont}
\usepackage{subcaption}
\usepackage{xcolor}
\usepackage{booktabs}
\usepackage[table]{xcolor}

\newtheorem{theorem}{Theorem}[section]

\newtheorem{assumption}{Assumption}

\newtheorem{remark}[theorem]{Remark}
\newtheorem{problem}[theorem]{Problem}
\newenvironment{proof}{\par\noindent\textbf{Proof.} }{\hfill$\blacksquare$\par}

\title{ Glycemic Safety Tube: A Provably Safe Control Framework for Artificial Pancreas Systems under Parametric Uncertainty               
\thanks{ This work was supported through the Kotak IISc AI-ML Centre (KIAC), Indian Institute of Science, Bengaluru, with support from Google.}
}

\author{
 Pukhrambam Akash Singh \\
  Robert Bosch Centre for Cyber-Physical Systems\\
  IISc, Bengaluru, India\\
  \texttt{pukhrambams@iisc.ac.in} \\
   \And
 Ratnangshu Das \\
  Robert Bosch Centre for Cyber-Physical Systems\\
  IISc, Bengaluru, India\\
  \texttt{ratnangshud@iisc.ac.in} \\
  \And
 Ahan Basu \\
  Robert Bosch Centre for Cyber-Physical Systems\\
  IISc, Bengaluru, India\\
  \texttt{ahanbasu@iisc.ac.in} \\
   \And
 Pushpak Jagtap \\
  Robert Bosch Centre for Cyber-Physical Systems\\
  IISc, Bengaluru, India\\
  \texttt{pushpak@iisc.ac.in} \\
}

\begin{document}
\maketitle

\begin{abstract}
Type~1 diabetes eliminates the body’s ability to produce insulin, making glucose regulation entirely dependent on external insulin delivery and the control algorithm. Existing closed-loop methods either rely on accurate patient-specific models or do not provide formal safety guarantees, and are often computationally demanding for wearable devices.
This paper proposes Glycemic Safety Tube Control (GSTC), a model-free and computationally efficient control framework for automated insulin delivery. The method enforces clinically relevant safety bounds on glucose levels by design, ensuring that glucose remains within a prescribed safe range. We also derive feasibility conditions that guarantee safety and input constraint satisfaction under bounded meal disturbances and estimation errors.
The performance of GSTC is evaluated against state-of-the-art methods, including linear and nonlinear model predictive control and sliding mode control. The results demonstrate that GSTC maintains safety under varying meal patterns and patient conditions, highlighting its robustness and computational efficiency. Overall, GSTC provides a safe, efficient, and patient-independent approach for next-generation artificial pancreas systems.
\end{abstract}

\keywords{
Artificial Pancreas Systems; Automated Insulin Delivery; Blood Glucose Control; Parametric Uncertainty; Formal Guarantees; Safety-Critical Control; Real-time Implementation}

\section{Introduction}
\label{sec:intro}
Type~1 diabetes mellitus (T1D) is a chronic autoimmune condition in which destruction of pancreatic $\beta$-cells eliminates endogenous insulin secretion entirely, requiring lifelong exogenous delivery for survival~\cite{atkinson2001type1}. As reported by \cite{IDF2024Atlas}, the International Diabetes Federation estimates 589~million adults currently live with diabetes, with T1D being the predominant form amongst children and adolescents, and projections reaching 853~million by 2050. Conventional management imposes a continuous cognitive burden on patients~\cite{peyrot2012insulin}: underdosing drives chronic hyperglycemia with vascular, renal, and neurological consequences~\cite{dcct1993intensive}, while overdosing precipitates hypoglycemia that may progress into hypoglycemic coma or death~\cite{cameron2011risk}.

\begin{figure}
    \centering
    \includegraphics[width=0.6\linewidth]{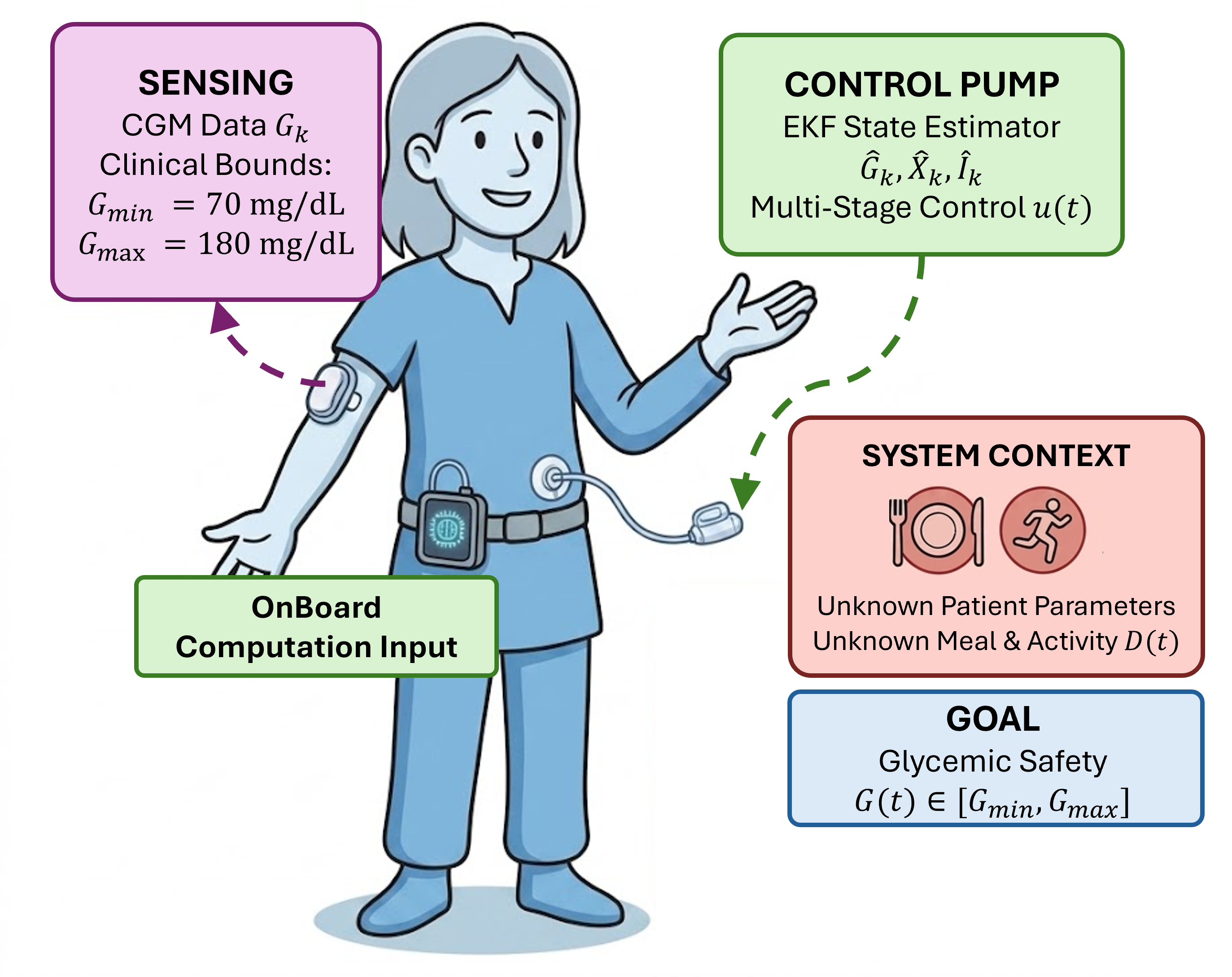}
    \caption{Schematic representation of the Automated Insulin Delivery system, illustrating the physical integration of the Continuous Glucose Monitor (CGM) and insulin pump with the Glycemic Safety Tube Control (GSTC) architecture.}
    \label{fig:AID}
\end{figure}

The Artificial Pancreas (AP) closes the loop between a Continuous Glucose Monitor (CGM) and a subcutaneous insulin pump through an embedded control algorithm~\cite{Doyle2025twenty}. \cite{Battelino2019} established the internationally agreed clinical target, that requires at least 70\% of readings within 70--180~mg/dL, with strict limits on time below 70~mg/dL. Commercial hybrid closed-loop systems have demonstrated meaningful progress~\cite{zhong2016minimed640g, pinsker2020basaliq}, and their real-world adoption continues to grow~\cite{Considine2024realworld}. However, fully autonomous glucose control in free-living conditions remains 
a challenging problem~\cite{Russell2014bionic}.

The risk profile in glucose regulation is asymmetric: hypoglycemia must be strictly avoided and is typically prioritized over tight upper-range control~\cite{kovatchev1997symmetrization,cameron2011risk}. In practice, unannounced meals can cause sudden increases in glucose that the controller cannot anticipate. At the same time, insulin sensitivity varies not only across patients, but also within the same individual over the course of a day, as reported by~\cite{hovorka2004nonlinear}. These factors make the system uncertain and difficult to model accurately. 
As a result, the controller must maintain safety within strict limits while operating without patient-specific calibration or meal information.

Several control approaches have been proposed for automated insulin delivery. PID control is simple but lacks physiological awareness and cannot handle input constraints, often leading to insulin over-delivery during meals~\cite{steil2004closedloop}. Model-based methods such as MPC~\cite{magni2007mpc} and CBF~\cite{saxena2025safe} can handle constraints more systematically, but they rely heavily on accurate patient-specific models, which are nonlinear, time-varying, and difficult to obtain in practice~\cite{Wilinska2010}. Observer-based~\cite{nath2019observer,Targui2024observer,Hooshmandi2024LPV} and learning-based~\cite{Daskalaki2013,fox2020deep} approaches improve adaptability to uncertainty and patient variability, but they generally lack formal safety guarantees, which is critical in medical applications. In addition, many of these methods are computationally expensive, making real-time implementation on wearable devices challenging. Overall, no existing approach simultaneously provides model-free operation, formal safety guarantees, and computational efficiency.

Prescribed performance control (PPC) by \cite{bechlioulis2014low} and the spatiotemporal tubes (STT) framework by \cite{das2025spatiotemporal} provide a promising alternative to address these challenges. These approaches enable approximation-free, closed-form control laws that constrain system trajectories within predefined bounds~\cite{RDas2024,das2025real}, while remaining computationally efficient and not requiring explicit knowledge of system parameters. However, a key limitation is that they do not inherently account for input constraints. Recent work by~\cite{das2025approximation} addresses this issue by introducing modified transformation functions, along with feasibility conditions, to ensure constraint satisfaction.

Building on these advances in approximation-free control, we adapt the framework to the glucose regulation setting through the proposed \emph{Glycemic Safety Tube} (GST). The GST captures clinically established safety bounds on glucose evolution and enforces them by design. 
Unlike the systems considered by~\cite{das2025approximation}, the glucose--insulin dynamics exhibit a cascaded structure with coupled subsystems corresponding to glucose, remote insulin action, and plasma insulin. This requires a hierarchical control design. In addition, insulin infusion is inherently one-sided (non-negative), so standard symmetric control transformations cannot be used. To address this, we develop new transformation functions along with corresponding feasibility conditions that ensure the safe evolution of the glucose level. The resulting GST-based controller is model-free, computationally lightweight, and does not require patient-specific identification, online optimization, or learning-based approximation, relying only on the basal glucose value.

This paper makes the following contributions:
\begin{itemize}
    \item We introduce Glycemic Safety Tubes (GST), which encode clinically meaningful safety bounds and ensure that glucose remains within these bounds.    
    
    \item We propose a cascaded, closed-form control framework for the glucose--insulin system that explicitly accounts for its three-layer structure: glucose, remote insulin action, and plasma insulin.

    \item We design new transformation functions and derive feasibility conditions to handle asymmetric input constraints, ensuring that actuator limits are respected while maintaining safety under worst-case meal disturbances and estimation errors.

    \item We provide formal guarantees that glucose remains within safe limits and actuator constraints are satisfied, for all patients within a given uncertainty set, without requiring patient-specific model identification.

    \item We demonstrate the effectiveness of the proposed approach through simulations on both the Bergman minimal model and a high-fidelity 13-state preclinical simulator for different patient and meal conditions. We also compare the results with benchmark controllers.
\end{itemize}


\section{System Description and Problem Formulation}
This section presents the glucose-insulin interaction model, states the assumptions on system uncertainty, and formally defines the control problem addressed in this work.

\subsection{System Model}
\label{sec:model}
The Bergman Minimal Model~\cite{Bergman1981, Fisher68209}, describes glucose--insulin dynamics using three states: the blood glucose deviation $G(t)$ from its basal value $G_b$, the remote insulin action $X(t)$, and the plasma insulin deviation $I(t)$ from its basal value $I_b$. The system dynamics are given by:
\begin{align}\label{eqn:sysdyn}
    \dot{G} = -S_G G - X(G + G_b) + D, \quad
    \dot{X} = - P_2 X + P_3 I, \quad
    \dot{I} = -n(I + I_b) + {u}/{V}, 
\end{align}
where $S_G, P_2, P_3,$ and $n$ are patient-specific parameters, $V$ is the insulin distribution volume (see Table~\ref{tab:params}). $u(t)$ is the insulin infusion rate (control input).
While the structure of these dynamics is fixed, the parameters $S_G, P_2, P_3,$ and $n$ vary across and within patients, and the meal input $D(t)$ is unknown in practice. The following assumptions formalize the extent of this uncertainty and form the basis on which formal safety guarantees are established.

\begin{assumption}\label{assum:param_bound}
The patient-specific parameters are uncertain but bounded. In particular,
$$
S_G \in [\underline{S}_G, \overline{S}_G], \quad
P_2 \in [\underline{P}_2, \overline{P}_2], \quad
P_3 \in [\underline{P}_3, \overline{P}_3], \quad
n \in [\underline{n}, \overline{n}],
$$
where all bounds are known positive constants.
\end{assumption}
This assumption reflects inter- and intra-patient variability in physiological parameters. While the exact values are unknown, their ranges can be estimated from clinical data and prior studies~\cite{nath2019observer}. This enables robust control design without requiring precise patient identification.

\begin{assumption}\label{assum:meal_bound}
The exogenous meal disturbance $D(t)$ is unknown and time-varying, but bounded as
$$
0 \leq D(t) \leq \overline{D}.
$$
\end{assumption}
This assumption captures the uncertainty in meal timing and carbohydrate intake, which are typically not known a priori in practical settings. Only a bound on the disturbance magnitude is assumed, allowing the controller to operate without explicit meal prediction or announcement.

The control input (insulin infusion, mU/min) $u(t)$ is subject to actuator constraints
$$
u(t) \in [0, \overline{u}].
$$

This setup captures key practical limitations in artificial pancreas systems. The exact patient parameters are unknown due to physiological variability, and meal disturbances are not known a priori. Additionally, the insulin infusion input is constrained by pump limitations. The control design must therefore ensure safety and performance under bounded uncertainty without relying on precise model knowledge or disturbance prediction.

\begin{table*}[t]
\caption{Population-average parameters of the Bergman Minimal Model~\cite{10.2337/diacare.8.6.553}.}
\label{tab:params}
\centering
\begin{tabular}{l l c l}
\hline
\textbf{Parameter} & \textbf{Description} & \textbf{Value} & \textbf{Units} \\
\hline
$S_G$ & Glucose effectiveness & $0.028$ & min$^{-1}$ \\
$P_2$ & Remote insulin rate constant & $0.025$ & min$^{-1}$ \\
$P_3$ & Insulin--glucose coupling & $1.3 \times 10^{-5}$ & (mU/L)$^{-1}$ min$^{-2}$ \\
$G_b$ & Basal glucose level & patient-specific & mg/dL \\
$I_b$ & Basal plasma insulin & patient-specific & mU/mL \\
$n$   & Insulin clearance rate & $5/54$ & min$^{-1}$ \\
$V$   & Insulin distribution volume & $12$ & L \\
\hline
\end{tabular}
\end{table*}

    
    


\subsection{Control Objective}

The goal is to maintain glucose $G(t) \in [\underline{G}, \overline{G}]$ for all $t > 0$,  under unannounced meals and patient parameter uncertainty. The control problem is formally stated as follows:

\begin{problem}\label{prob}
    Given the system~\eqref{eqn:sysdyn} and bounds $[\underline{G}, \overline{G}]$, design a model-free feedback control law $u(t)$, such that: 
    \begin{itemize}
        \item[(i)] the glucose level remains strictly within the safe range 
        $$G(t) \in [\underline{G}, \overline{G}], \quad \forall t > 0,$$
        \item[(ii)] the insulin infusion pump limit is constrained 
        $$u(t) \in [0,\overline{u}], \quad \forall t > 0.$$
    \end{itemize}
\end{problem}

Problem~\ref{prob} is stated in terms of the true glucose state $G(t)$, but in practice, only a noisy glucose measurement is available through the CGM; the internal states $X(t)$ and $I(t)$ are entirely inaccessible to direct measurement. The controller must therefore operate on reconstructed state estimates, and those estimates must come with certified error bounds that can be incorporated into the control design as hard guarantees rather than probabilistic assumptions.  

\subsection{State Estimation and Uncertainty Characterization}
\label{subsec:state_est}

The states $X$ and $I$ in~\eqref{eqn:sysdyn} are not directly measurable; only glucose is accessible via CGM~\cite{Boiroux2017}. These fundamental observability constraints necessitate state estimation~\cite{parker1999modelbased, hovorka2004nonlinear}. The bilinear coupling in~\eqref{eqn:sysdyn} motivates the use of a standard Extended Kalman Filter~\cite{Simon10.5555/1146304} to reconstruct $\hat{\mathbf{x}}_k = [\hat{G}_k,\,\hat{X}_k,\,\hat{I}_k]^\top$ via two steps at each instant: a prediction step that propagates the model forward, and an update step that fuses the CGM measurement. The superscript $(-)$ denotes the a priori estimate produced by the prediction step, before measurement fusion. Discretizing~\eqref{eqn:sysdyn} with sampling time ($T_s$) gives:
\begin{align}
  \hat{G}_{k+1}^{-} &= \hat{G}_k + T_s\!\left[
      -S_G\hat{G}_k - \hat{X}_k(\hat{G}_k+G_b)\right], 
      \label{eq:Gpred}\\
  \hat{X}_{k+1}^{-} &= \hat{X}_k + T_s\!\left[
      P_3\hat{I}_k - P_2\hat{X}_k\right], \label{eq:Xpred}\\
  \hat{I}_{k+1}^{-} &= \hat{I}_k + T_s\!\left[
      -n(\hat{I}_k+I_b) + \frac{u_k}{V_I}\right], \label{eq:Ipred}
\end{align}
The estimator uses nominal (population-average) parameter values for $S_G$, $P_2$, $P_3$, $n$, and $V$, given in table~\ref{tab:params}. The true patient parameters may differ from these values, and this mismatch is accounted for through the controller designed in the next section.

The discrete Jacobian is
\begin{equation}
  F_k = \begin{bmatrix}
    1-T_s(S_G+\hat{X}_k) & -T_s(\hat{G}_k+G_b) & 0 \\[4pt]
    0 & 1-T_s P_2 & T_s P_3 \\[4pt]
    0 & 0 & 1-T_s n
  \end{bmatrix},
\label{eq:Fk}
\end{equation}
where propagating the error covariance as $\Sigma_{k+1}^{-} = F_k \Sigma_k F_k^\top + Q$. The innovation $(z_k - G_b) - H\hat{\mathbf{x}}_k^{-}$, where $H = [1\ 0\ 0]$ and $z_k$ is the absolute CGM reading, drives the Joseph-form update~\cite{Simon10.5555/1146304} to yield $\Sigma_{k|k}$.

After the filter converges, $\Sigma_{k|k} \to \Sigma^*$ is approximately stationary; its diagonal yields the $3\sigma$ state estimation bounds:
\begin{equation}
  [\Delta_G, \Delta_X, \Delta_I]^\top = 3\sqrt{\mathrm{diag}(\Sigma^*)}.
\label{eq:state_bounds}
\end{equation}

The derivative bounds follow from the error dynamics. 
\begin{align}\label{eq:dotD}
  \dot{\Delta}_I &:= \overline{n}\,\Delta_I, \quad
  \dot{\Delta}_X := \overline{P}_3\,\Delta_I + \overline{P}_2\,\Delta_X, \quad
  \dot{\Delta}_G := (\overline{S}_G+\overline{X})\,\Delta_G
               + (\overline{G}+G_b)\,\Delta_X + \overline{D}. 
\end{align}
Here $\overline{X}$ is the upper bound on $X$. 
Although state estimates are updated at discrete sampling instants, the control law operates in continuous time. The estimated states are held constant between sampling instants using a zero-order hold. 

The position bounds~\eqref{eq:state_bounds} define the uncertainty radius around each state estimate; the derivative bounds~\eqref{eq:dotD} define the worst-case rate at which that uncertainty can evolve. Together, they capture the uncertainty characterization that the controller of Section~\ref{sec:control} is designed to dominate at every stage. No further knowledge of the patient's physiology is assumed beyond what these six constants encode.

\section{Control Design} \label{sec:control}

The control law solving Problem~\ref{prob} follows a three-stage backstepping-inspired procedure, in the spirit of~\cite{bechlioulis2014low, das2025spatiotemporal}: reference signals for the internal states are generated progressively, ultimately yielding the physical control input.

Specifically, we first design a reference remote insulin signal $X_{ref}$ to regulate glucose within the safe range, and design a reference plasma insulin $I_{ref}$ to track $X_{ref}$. Finally, we design the insulin infusion input $u(t)$ to track $I_{ref}$. 

\subsection{Bounded Transformation Function} \label{sec:clamp}

The function $\Psi: \mathbb{R} \rightarrow [0,1]$ is a continuously differentiable, nondecreasing mapping used to enforce bounded control action. It is defined such that:
\begin{equation}
    \Psi(s) = \begin{cases} 0, & s \in (-\infty,0], \\ 1, & s \in [1,\infty), \end{cases} 
\end{equation}
and $\Psi(s)$ is nondecreasing. An example is shown in Figure~\ref{fig:sat}.

\begin{figure}
    \centering
    \includegraphics[width=0.4\linewidth]{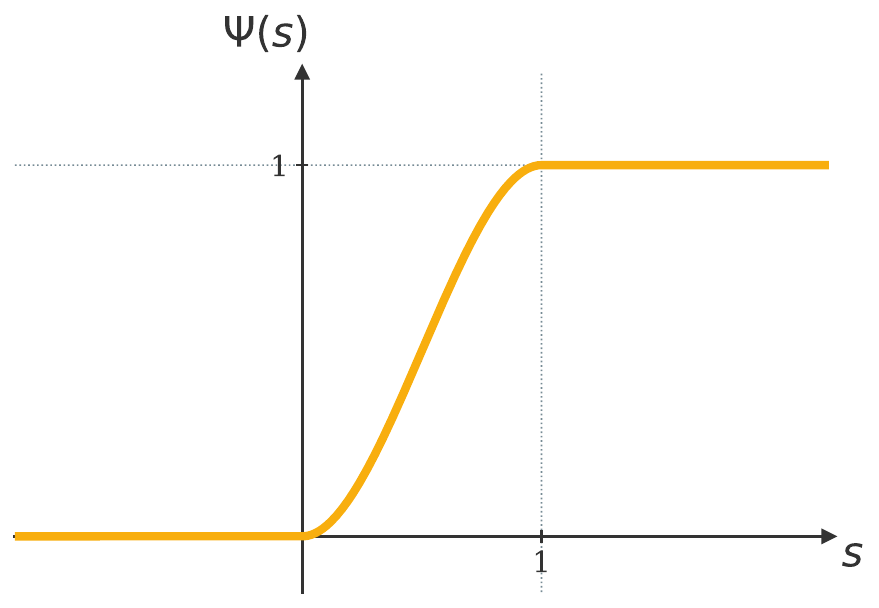}
    \caption{The bounded transformation function $\Psi(s)$, mapping the normalised error to $[0,1]$. }
    \label{fig:sat}
\end{figure}

In particular, when $s \ge 1$, the function saturates at $1$, enforcing the maximum allowable input. When $s \leq 0$, the output is $0$, corresponding to no insulin infusion.

\subsection{Stage I: Glucose Regulation} \label{subsec:stage1}

Given the safe glucose interval $[\underline{G}, \overline{G}]$, define the normalized error:
\begin{align}
    e_G(t) = \frac{G(t) - 0.5(\overline{G} + \underline{G})}{0.5(\overline{G} - \underline{G})}.
\end{align}

This normalization maps the safe set to $e_G \in [-1,1]$. The control objective is therefore equivalent to maintaining $|e_G(t)| < 1$.
We define the reference remote insulin as:
\begin{align}\label{eqn:Xref}
    X_{ref}(t) = \kappa_1 \Psi(e_G(t)),
\end{align}
where $\Psi(\cdot)$ is the bounded transformation defined in sub-section~\ref{sec:clamp}, and $\kappa_1 \in \mathbb{R}^+$ is a tunable gain.

This construction ensures that larger deviations in glucose lead to increased insulin action, while keeping the reference bounded. The glucose safety specification is now encoded entirely in $X_{ref}(t)$  . However, $X_{ref}(t)$ is a desired trajectory for a physiological state. Stage~II now derives the plasma insulin level $I_{ref}(t)$  needed to drive the remote compartment reliably onto this reference.

\subsection{Stage II: Tracking of Remote Insulin}

To track $X_{ref}(t)$, define the tracking error:
\begin{align}
    \hat{e}_X(t) = X_{ref}(t) - \hat{X}(t).
\end{align}

We enforce this error to remain within a time-varying funnel, defined by $\rho_X(t): \mathbb{R}_0^+ \to \mathbb{R}^+$:
\begin{equation}\label{eqn:rhov}
    \rho_X(t) = e^{-\mu_X t}(p_X - q_X) + q_X,
\end{equation}
where $p_X \in \mathbb{R}^+$ is the initial bound satisfying $|\hat{e}_X(0)| \leq p_X$, $q_X \in \mathbb{R}^+$ is the steady-state bound with $0 < q_X < p_X$, and $\mu_X > 0$ is the convergence rate.
The funnel constraint is $-\rho_X(t) < \hat{e}_X(t) < \rho_X(t).$
Define the normalized error:
\begin{equation}\label{eqn:norm_err2}
    e_X(t) = \frac{\hat{e}_X(t)}{\rho_X(t)}.
\end{equation}

The reference plasma insulin is then given by:
\begin{align}\label{eqn:Iref}
    I_{ref}(t) = \kappa_2 \Psi(e_X(t)),
\end{align}
where $\kappa_2 \in \mathbb{R}^+$ is a design parameter.

This guarantees that the tracking error remains within the prescribed funnel and gradually decreases, ensuring accurate tracking. Plasma insulin $I_{ref}(t)$ specifies the target, but it too is a physiological state rather than a control input. Stage~III closes this gap, mapping $I_{ref}(t)$ into the physically deliverable infusion rate $u(t)$ that the pump can execute.

\subsection{Stage III: Tracking of Plasma Insulin}

To track $I_{ref}(t)$, define the tracking error:
\begin{align}
    \hat{e}_I(t) = I_{ref}(t) - \hat{I}(t).
\end{align}

We impose a funnel constraint using $\rho_I(t): \mathbb{R}_0^+ \to \mathbb{R}^+$:
\begin{equation}
    \rho_I(t) = e^{-\mu_I t}(p_I - q_I) + q_I,
\end{equation}
where $p_I \in \mathbb{R}^+$ satisfies $|\hat{e}_I(0)| \leq p_I$, $q_I \in \mathbb{R}^+$ with $0 < q_I < p_I$, and $\mu_I > 0$.
The funnel constraint is $-\rho_I(t) < \hat{e}_I(t) < \rho_I(t).$
Define the normalized error:
\begin{equation}
    e_I(t) = \frac{\hat{e}_I(t)}{\rho_I(t)}.
\end{equation}

The control input is then given by:
\begin{align}\label{eqn:insulin_control}
    u(t) = \overline{u} \Psi(e_I(t)),
\end{align}
where $\overline{u} \in \mathbb{R}^+$ is the pump limit.

The three control laws together constitute a complete closed-form controller requiring neither online optimization nor patient-specific identification. Whether they collectively maintain the glucose safety tube under worst-case uncertainty depends on a set of explicit conditions on the design parameters, derived in the following subsection.

\subsection{Feasibility Conditions} \label{subsec:feasibility}

To ensure safe operation under worst-case disturbances and estimation errors, the following conditions must hold:

\begin{enumerate}
    \item \textbf{Glucose regulation condition:}
    \begin{align}
        (\kappa_1 - p_X - \Delta_X)(\overline{G} + G_b) > 
        \overline{D}\tag{FC1a}\label{eq:feas1} \\    
        -\underline{S}_G \underline{G} > (p_X + \Delta_X)
        (\underline{G} + G_b) \label{eq:feas1.5}\tag{FC1b}
    \end{align}

    \textit{Interpretation:} \ref{eq:feas1} bounds the gain $\kappa_1$ from below: the net insulin action available at peak glycemic excursion must exceed the worst-case meal disturbance. \ref{eq:feas1.5} governs the lower boundary where the controller has driven insulin demand to zero and the body's intrinsic glucose effectiveness must be strong enough to overcome the residual insulin action left. 

    \item \textbf{Insulin action tracking condition:}
    \begin{align}
        &\underline{P}_3(\kappa_2 - p_I - \Delta_I) - 
        \overline{P}_2(\kappa_1 - q_X + \Delta_X) > \sup_t |\dot{X}_{\mathrm{ref}}| + 
        \mu_X(p_X - q_X) + 
        \dot{\Delta}_X\label{eq:feas2} \tag{FC2a} \\
        &\underline{P}_2(q_X - \Delta_X) - 
        \overline{P}_3(p_I + \Delta_I) > \sup_t|\dot{X}_{\mathrm{ref}}| 
           + \mu_X(p_X - q_X) + 
           \dot{\Delta}_X\label{eq:feas2.5}\tag{FC2b}
    \end{align}

    \textit{Interpretation:} \ref{eq:feas2} requires that plasma insulin can drive remote insulin action upward fast enough to track $X_{\mathrm{ref}}$ against worst-case parameter uncertainty and coupling from the glucose loop. 
    \ref{eq:feas2.5} requires that natural insulin decay is fast enough to drive remote action downward when demanded. Since the pump cannot withdraw insulin, clearance is the only reduction mechanism. 

    \item \textbf{Actuation condition:}
    \begin{align}
        &\frac{\overline{u}}{V} - \overline{n}(\kappa_2 - 
        q_I + \Delta_I + I_b) > \sup_t |\dot{I}_{\mathrm{ref}}| + 
        \mu_I(p_I - q_I) + 
        \dot{\Delta}_I \label{eq:feas3}\tag{FC3a} \\
        &\underline{n}(q_I - \Delta_I + I_b) > \sup_t |\dot{I}_{\mathrm{ref}}| + 
        \mu_I(p_I - q_I) + 
        \dot{\Delta}_I \label{eq:feas3.5}\tag{FC3b}
    \end{align}

    \textit{Interpretation:} \ref{eq:feas3} is a hardware constraint: the maximum pump rate $\overline{u}$ must cover the peak plasma insulin demand after accounting for natural clearance and estimation error. 
    \ref{eq:feas3.5} ensures that the natural clearance is fast enough to reduce plasma insulin at the required rate when insulin infusion is cut to zero. 

\end{enumerate}
Each condition binds a specific design parameter, $\kappa_1$, $\kappa_2$, and $\overline{u}$, to the level of uncertainty it must overcome. 

\textit{While these requirements align with intuitive clinical understanding, such as ensuring sufficient insulin to counter meals and adequate clearance to avoid accumulation, they are formalized here as explicit feasibility conditions. 
} 

The following theorem formalizes these results.

\subsection{Main Result}

\begin{theorem}\label{thm:glucose_control}
Consider the glucose--insulin system in \eqref{eqn:sysdyn} under bounded meal disturbance $D(t) \leq \overline{D}$ and bounded estimation errors. Suppose the initial conditions satisfy
$$|e_G(0)| < 1, \quad |e_X(0)| < 1, \quad |e_I(0)| < 1,$$
and the feasibility conditions (\ref{eq:feas1}-~\ref{eq:feas3.5}) hold. Then, the proposed control law \eqref{eqn:insulin_control} guarantees that for all $t \geq 0$ :
\begin{itemize}
    \item[(i)] Glucose evolves within a safe range
    $G(t) \in [\underline{G}, \overline{G}]$
    \item[(ii)] the insulin infusion control is bounded
    $u(t) \in [0, \overline{u}],$
\end{itemize}
\end{theorem}

\begin{proof}
The proof proceeds in three stages corresponding to the 
cascaded structure of the controller. We first note that 
by definition of the funnel $\rho(t) = (p-q)e^{-\mu t}+q$, 
its derivative satisfies $\dot{\rho}(t) = -\mu(p-q)e^{-\mu t}$, 
so that $\sup_t(-\dot{\rho}(t)) = \mu(p-q)$. This identity 
is used in each stage below to connect the feasibility 
conditions to the funnel derivative bounds.

\textbf{Stage I: Invariance of the $I$-funnel.}

We prove that $|e_I(t)| < 1$ for all $t \ge 0$ by 
contradiction. Assume there exists a first time $\tau_I > 0$ 
such that
\begin{align}\label{eq:contradiction_I}
    |e_I(\tau_I)| = 1 \implies |e_I(t)| < 1, 
    \ \forall\, t \in [0,\tau_I).
\end{align}

\textbf{Upper boundary ($e_I \to 1$):}
Since $e_I(t) < 1$ for all $t < \tau_I$, reaching 
$e_I(\tau_I)=1$ from below requires
\begin{equation}\label{eqn:stage1_contradict_up}\tag{C-I-U}
    \dot{I}_{ref} - \dot{\hat{I}} > \dot{\rho}_I 
    \implies \dot{\hat{I}} < \dot{I}_{ref} - \dot{\rho}_I.
\end{equation}
At $e_I = 1$ we have $\hat{I} = I_{\mathrm{ref}} - \rho_I$, 
hence, using the funnel and estimation bounds,
\begin{equation*}
    I \le \hat{I} + \Delta_I 
      = I_{\mathrm{ref}} - \rho_I + \Delta_I 
      \le \kappa_2 - q_I + \Delta_I,
\end{equation*}
as $I_{\mathrm{ref}} \le \kappa_2$ and $\rho_I \ge q_I$. 
As $e_I \to 1$, the control saturates to $u \to \overline{u}$. 
Substituting the upper bound on $I$ into the dynamics \eqref{eqn:sysdyn} and applying feasibility condition \eqref{eq:feas3}, we obtain
\begin{align*}
\dot{I} &\ge -\overline{n}\!\left(\kappa_2 - q_I 
          + \Delta_I + I_b\right) 
          + \frac{\overline{u}}{V} \quad (\because -n \geq -\overline{n}) \\
    \dot{I} &> \sup_t|\dot{I}_{ref}| + \mu_I(p_I-q_I) + \dot{\Delta}_I \\
    \dot{\hat{I}} &\ge \dot{I} - \dot{\Delta}_I > \sup_t|\dot{I}_{ref}| + \mu_I(p_I-q_I)\\    
    &\implies \dot{\hat{I}} > \dot{I}_{ref}(t) - \dot{\rho}_I,
\end{align*}
which contradicts \eqref{eqn:stage1_contradict_up}. 
Hence $e_I(t)$ cannot reach $1$.

\textbf{Lower boundary ($e_I \to -1$):}
Since $e_I(t) > -1$ for all $t < \tau_I$, reaching 
$e_I(\tau_I)=-1$ from above requires
\begin{equation}\label{eqn:stage1_contradict_low}\tag{C-I-L}
    \dot{I}_{ref} - \dot{\hat{I}} < -\dot{\rho}_I 
    \implies \dot{\hat{I}} > \dot{I}_{ref} + \dot{\rho}_I.
\end{equation}
At $e_I = -1$ we have $\hat{I} = I_{\mathrm{ref}} + \rho_I$, 
hence, using the funnel and estimation bounds,
\begin{equation*}
    I \ge \hat{I} - \Delta_I \ge q_I - \Delta_I,
\end{equation*}
as $I_{\mathrm{ref}} \ge 0$ and $\rho_I \ge q_I$. 
As $e_I \to -1$ the control cuts to $u \to 0$. 
Substituting the lower bound on $I$ into the dynamics and applying feasibility condition \eqref{eq:feas3.5}, we obtain
\begin{align*}
    \dot{I} &\le -\underline{n}(q_I - \Delta_I + I_b) \quad (\because -n \leq -\underline{n})\\
    \dot{I} &< -\sup_t|\dot{I}_{ref}| - \mu_I(p_I-q_I) - \dot{\Delta}_I \\
    \dot{\hat{I}} &\le \dot{I} + \dot{\Delta}_I < -\sup_t|\dot{I}_{ref}| - \mu_I(p_I-q_I) \\
    &\implies \dot{\hat{I}} < \dot{I}_{ref}(t) + \dot{\rho}_I,
\end{align*}
which contradicts \eqref{eqn:stage1_contradict_low}. 
Hence $e_I(t)$ cannot reach $-1$.

Thus $e_I(t)$ cannot reach $1$ or $-1$, i.e., $|e_I(t)| < 1$, and equivalently $|I_{ref}(t) - \hat{I}(t)| < \rho_I(t)$, for all $t \ge 0$.

\vspace{0.2cm}
\textbf{Stage II: Invariance of the $X$-funnel.}

We prove that $|e_X(t)| < 1$ for all $t \ge 0$ by 
contradiction. Assume there exists a first time $\tau_X > 0$ 
such that
\begin{align}\label{eq:contradiction_X}
    |e_X(\tau_X)| = 1 \implies |e_X(t)| < 1, 
    \ \forall\, t \in [0,\tau_X).
\end{align}

\textbf{Upper boundary ($e_X \to 1$):}
Since $e_X(t) < 1$ for all $t < \tau_X$, reaching 
$e_X(\tau_X)=1$ from below requires
\begin{equation}\label{eqn:stage2_contradict_up}\tag{C-II-U}
    \dot{X}_{ref} - \dot{\hat{X}} > \dot{\rho}_X 
    \implies \dot{\hat{X}} < \dot{X}_{ref} - \dot{\rho}_X.
\end{equation}
As $e_X \to 1$, from \eqref{eqn:Iref} we have 
$I_{ref} \to \kappa_2$. Using the funnel constraint 
$|\hat{I}-I_{ref}| < \rho_I(t) \le p_I$ (established in 
Stage~I) and the estimation bound $|I - \hat{I}| \le \Delta_I$:
\begin{equation*}
    I \ge I_{ref} - \rho_I - \Delta_I 
      \ge \kappa_2 - p_I - \Delta_I.
\end{equation*}
At $e_X = 1$ we have $\hat{X} = X_{ref} - \rho_X$, hence
\begin{equation*}
    X \le \hat{X} + \Delta_X 
      = X_{ref} - \rho_X + \Delta_X 
      \le \kappa_1 - q_X + \Delta_X,
\end{equation*}
using $X_{ref} \le \kappa_1$ and $\rho_X \ge q_X$. 
Substituting both bounds into the dynamics \eqref{eqn:sysdyn}, and applying feasibility condition \eqref{eq:feas2}, we obtain 
\begin{align*}
    \dot{X} &\ge \underline{P}_3(\kappa_2 - p_I - \Delta_I) 
             - \overline{P}_2(\kappa_1 - q_X + \Delta_X) 
             \\ &\hspace{3.4cm} (\because \ P_3 \ge \underline{P}_3, -P_2 \ge -\overline{P}_2) \\
    \dot{X} &> \sup_t|\dot{X}_{ref}| + \mu_X(p_X-q_X) + \dot{\Delta}_X   \\
    \dot{\hat{X}} &\ge \dot{X} - \dot{\Delta}_X > \sup_t|\dot{X}_{ref}| + \mu_X(p_X-q_X) \\
    &\implies \dot{\hat{X}} > \dot{X}_{ref}(t) - \dot{\rho}_X,
\end{align*}
which contradicts \eqref{eqn:stage2_contradict_up}. 
Hence $e_X(t)$ cannot reach $1$.

\textbf{Lower boundary ($e_X \to -1$):}
Since $e_X(t) > -1$ for all $t < \tau_X$, reaching 
$e_X(\tau_X)=-1$ from above requires
\begin{equation}\label{eqn:stage2_contradict_low}\tag{C-II-L}
    \dot{X}_{ref} - \dot{\hat{X}} < -\dot{\rho}_X 
    \implies \dot{\hat{X}} > \dot{X}_{ref} + \dot{\rho}_X.
\end{equation}
As $e_X \to -1$, from \eqref{eqn:Iref} we have 
$I_{ref} \to 0$. Using the funnel constraint 
$|\hat{I}-I_{ref}| < \rho_I(t) \le p_I$ (established in 
Stage~I) and the estimation bound $|I-\hat{I}| \le \Delta_I$:
\begin{equation*}
    I \le I_{ref} + \rho_I + \Delta_I \le p_I + \Delta_I.
\end{equation*}
At $e_X = -1$ we have $\hat{X} = X_{ref} + \rho_X$, hence
\begin{equation*}
    X \ge \hat{X} - \Delta_X \ge q_X - \Delta_X,
\end{equation*}
using $X_{ref} \ge 0$ and $\rho_X \ge q_X$. 
Substituting both bounds into the dynamics \eqref{eqn:sysdyn} and applying feasibility condition \eqref{eq:feas2.5}, we obtain
\begin{align*}
    \dot{X} &\le \overline{P}_3(p_I + \Delta_I) - \underline{P}_2(q_X - \Delta_X) \\
    & \hspace{3.2cm} (\because \ P_3 \le \overline{P}_3, -P_2 \le -\underline{P}_2) \\
    \dot{X} &< -\sup_t|\dot{X}_{ref}| - \mu_X(p_X-q_X) - \dot{\Delta}_X \\
    &\implies \dot{\hat{X}} < \dot{X}_{ref}(t) + \dot{\rho}_X,
\end{align*}
which contradicts \eqref{eqn:stage2_contradict_low}. 
Hence $e_X(t)$ cannot reach $-1$.

Thus $e_X(t)$ cannot reach $1$ or $-1$, i.e., $|e_X(t)| < 1$, and equivalently $|X_{ref}(t) - \hat{X}(t)| < \rho_X(t)$, for all $t \ge 0$.

\vspace{0.2cm}
\textbf{Stage III: Invariance of the glucose constraint.}

We now prove that $G(t) \in [\underline{G}, \overline{G}]$, 
equivalently $|e_G(t)| < 1$, for all $t \ge 0$ by 
contradiction. Assume there exists a first time $\tau_G > 0$ 
such that
\begin{align}\label{eq:contradiction_G}
    |e_G(\tau_G)| = 1 \implies |e_G(t)| < 1, 
    \ \forall\, t \in [0,\tau_G).
\end{align}

\textbf{Upper boundary ($G \to \overline{G}$):}
Since $G(t) < \overline{G}$ for all $t < \tau_G$, reaching 
$G(\tau_G) = \overline{G}$ from below requires
\begin{equation}\label{eqn:stage3_contradict_up}\tag{C-III-U}
    \lim_{G(t)\to\overline{G}} \dot{G}(t) \ge 0.
\end{equation}
As $G \to \overline{G}$, from \eqref{eqn:Xref} we have 
$X_{ref} \to \kappa_1$. Using the funnel constraint 
$|\hat{X}-X_{ref}| < \rho_X(t) \le p_X$ (established in 
Stage~II) and the estimation bound $|X-\hat{X}| \le \Delta_X$:
\begin{equation*}
    X \ge X_{ref} - \rho_X - \Delta_X 
      \ge \kappa_1 - p_X - \Delta_X.
\end{equation*}
Substituting this lower bound on $X$ and the upper bound 
$D \le \overline{D}$ into the dynamics \eqref{eqn:sysdyn} at $G = \overline{G}$:
\begin{equation*}
    \lim_{G\to\overline{G}} \dot{G} 
    \le -S_G\overline{G} 
       - (\kappa_1 - p_X - \Delta_X)(\overline{G}+G_b) 
       + \overline{D}.
\end{equation*}
By feasibility condition \eqref{eq:feas1}, and since $S_G \ge 0$ and $\overline{G} > 0$ the term 
$-S_G\overline{G} \le 0$, so the right-hand side satisfies:
\begin{equation*}
    -S_G\overline{G} 
    - (\kappa_1-p_X-\Delta_X)(\overline{G}+G_b) 
    + \overline{D} < 0.
\end{equation*}
Hence $\lim_{G\to\overline{G}}\dot{G} < 0$, which contradicts 
\eqref{eqn:stage3_contradict_up}. Therefore $G(t)$ cannot 
reach $\overline{G}$.

\textbf{Lower boundary ($G \to \underline{G}$):}
Since $G(t) > \underline{G}$ for all $t < \tau_G$, reaching 
$G(\tau_G) = \underline{G}$ from above requires
\begin{equation}\label{eqn:stage3_contradict_low}\tag{C-III-L}
    \lim_{G(t)\to\underline{G}} \dot{G}(t) \le 0.
\end{equation}
As $G \to \underline{G}$, from \eqref{eqn:Xref} we have 
$X_{ref} \to 0$, so $\Psi(e_G) \to 0$ and thus 
$I_{ref} \to 0$. Using the funnel constraint 
$|\hat{X}-X_{ref}| < \rho_X(t) \le p_X$ and estimation 
bound $|X-\hat{X}| \le \Delta_X$:
\begin{equation*}
    X \le X_{ref} + \rho_X + \Delta_X \le p_X + \Delta_X.
\end{equation*}
Substituting this upper bound on $X$ and the bound 
$D(t) \ge 0$ into the dynamics \eqref{eqn:sysdyn} at $G = \underline{G}$:
\begin{equation*}
    \lim_{G\to\underline{G}} \dot{G} 
    \ge -S_G\underline{G} 
       - (p_X+\Delta_X)(\underline{G}+G_b) 
       + D(t).
\end{equation*}
By feasibility condition \eqref{eq:feas1.5}, and since $S_G \ge \underline{S}_G$ and $\underline{G} < 0$, 
we have $-S_G\underline{G} \ge -\underline{S}_G\underline{G} 
> (p_X+\Delta_X)(\underline{G}+G_b)$. Also $D(t) \ge 0$. So, the right-hand side satisfies:
\begin{equation*}
    -S_G\underline{G} - (p_X+\Delta_X)(\underline{G}+G_b) > 0.
\end{equation*}
Hence, $\lim_{G\to\underline{G}}\dot{G} > 0$, which 
contradicts \eqref{eqn:stage3_contradict_low}. Therefore 
$G(t)$ cannot reach $\underline{G}$.

Thus $G(t)$ cannot reach the boundaries $\underline{G}$ or $\overline{G}$, i.e., $G(t) \in [\underline{G}, \overline{G}]$ for all $t \ge 0$.

\vspace{0.2cm}
Since no violation time $\tau_G$, $\tau_X$, or $\tau_I$ 
exists, we conclude:
\begin{equation*}
    |e_I(t)| < 1, \quad |e_X(t)| < 1, \quad |e_G(t)| < 1, 
    \qquad \forall\, t \ge 0.
\end{equation*}
Therefore $G(t) \in [\underline{G},\overline{G}]$ and 
$u(t) = \overline{u}\,\Psi(e_I(t)) \in [0,\overline{u}]$ 
for all $t \ge 0$, completing the proof.
\end{proof}

The safety guarantee holds for any patient whose physiology lies within the admitted uncertainty set (Assumption~\ref{assum:param_bound}) and for all bounded meal disturbances (Assumption~\ref{assum:meal_bound}), without requiring individual model identification.

Section~\ref{sec:results} evaluates how these guarantees translate into glycemic regulation performance across two complementary settings: a MATLAB study with the nominal model and a high-fidelity preclinical simulator.

\section{ In Silico Evaluation Results and Discussion} \label{sec:results}

\subsection{Simulation Setup and Baseline Controllers}
The proposed GSTC is compared against four established closed-loop glucose control strategies: a PID controller with a predictive Control Barrier Function safety layer (PID+CBF)~\cite{saxena2025safe}, a Sliding Mode Controller (SMC)~\cite{OROZCOLOPEZ2026106931}, a Linear MPC (L-MPC)~\cite{SORU2012118}, and a Nonlinear MPC (NMPC)~\cite{magni2007mpc}. 

All controllers are evaluated using the same three-day, ten-meal simulation protocol, where meals are unannounced and not provided to the controllers in advance. Meals are modeled as carbohydrate disturbances entering the glucose dynamics through the standard absorption model given in~\cite{hovorka2004nonlinear}. Meal sizes between 40--70~g represent typical breakfast, lunch, snack, and dinner events. The resulting disturbance $D(t)$ follows a smooth rise-and-decay absorption profile, capturing the gradual appearance of glucose in the bloodstream after food intake rather than an instantaneous impulse.

To assess robustness to model mismatch, the controllers are tested on two patients with different physiological parameters. The parameters of Patient~1 are used as nominal model parameters for all controllers. As a result, when evaluated on Patient~2, the performance reflects the impact of model mismatch.

Performance is evaluated using standard clinical glucose metrics, following~\cite{Battelino2019}, together with statistical metrics. Clinical metrics include time in range (TIR), time in tight range (TITR), time above range (TAR), and time below range (TBR), which quantify the percentage of time glucose remains within clinically relevant regions. Statistical metrics such as maximum glucose, minimum glucose, mean glucose, interquartile range (IQR), and coefficient of variation (CV) are used to assess glucose excursions and glycemic variability. In addition, the peak insulin infusion rate $\overline{u}$ and average computational time per control step are reported to evaluate actuator usage and real-time computational efficiency. 

Note that all glucose trajectories shown in the subsequent plots correspond to the absolute glucose values rather than the deviation variables. This allows direct visualization of the clinically prescribed glucose safety bounds of 70--180~mg/dL.


\subsection{Safety Tube Verification: Single Patient Study}

The basal glucose level is $G_b = 80$~mg/dL, with the safety tube defined by $\underline{G} = -10$~mg/dL and $\overline{G} = 100$~mg/dL, corresponding to clinically established absolute bounds 70--180~mg/dL. To validate the formal safety guarantee of Theorem~\ref{thm:glucose_control}, we simulate the closed-loop system over a 72-hour period covering three days with varying meal disturbances. 

Instead of using a fixed meal schedule, the three-day program is varied to better reflect realistic patient behavior. Day~1 follows a regular breakfast, lunch, and dinner pattern with meal sizes 50~g, 70~g, and 40~g. Day~2 includes four meal events of 50~g, 65~g, 25~g, and 45~g, including an afternoon snack. Day~3 features a delayed evening meal with intakes of 55~g, 70~g, and 50~g. This variation allows us to test the controller across a range of meal sizes and timings.

Figure~\ref{fig:benchmark} shows the glucose trajectory $G(t)$ and pump input $u(t)$ over the full horizon for Patient~1. The glucose level remains strictly within the safe range across all meal events and fasting periods, confirming Theorem~\ref{thm:glucose_control}(i). The largest deviation occurs after the 70~g lunch on Day~1, where glucose level reaches about 167~mg/dL, which is 87~mg/dL above basal but still within the upper limit, leaving a margin of 13~mg/dL. The lower bound is never approached, and no hypoglycemic events occur during the entire 72-hour period. During fasting periods, glucose smoothly returns to the basal level $G_b$.

Each meal disturbance triggers a rapid increase in the insulin input $u(t)$ to counter the rise in glucose. The input remains below the maximum limit $\overline{u} = 144$~mU/min at all times, satisfying Theorem~\ref{thm:glucose_control}(ii). 

Having verified that all three conditions of Theorem~\ref{thm:glucose_control} are satisfied under varied meal disturbances, we next compare the controller with existing AP strategies across two patients to evaluate performance and robustness under parameter variability.

\subsection{Comparative Evaluation Against Baseline Controllers}

The proposed GSTC controller is now compared with PID+CBF, SMC, L-MPC, and NMPC using a 72-hour, 10-meal simulation in which meals are unannounced to all controllers. Table~\ref{tab:benchmark} summarizes the glycemic and computational performance for both patients, while Figures~\ref{fig:benchmark} and~\ref{fig:worst_benchmark} show the corresponding glucose and insulin profiles for Patients~1 and~2.

For Patient~1, GSTC achieves a TIR of 100.00\%, with a peak post-meal glucose of 155.62~mg/dL and a CV of 25.01\%. As shown in Figure~\ref{fig:benchmark}, glucose remains within the safe range $[70,\,180]$~mg/dL throughout all meals and fasting periods. L-MPC achieves the same TIR but requires insulin inputs that repeatedly hit the maximum limit and has a much higher computational cost (1.047~ms compared to 0.00225~ms per step). NMPC, even with perfect model knowledge, achieves a lower TIR (96.46\%) and reaches a higher peak glucose of 197.89~mg/dL, highlighting the absence of formal guarantees on glucose safety and constraint satisfaction.. PID+CBF exceeds the upper limit during several meals, while SMC keeps glucose close to the lower boundary, resulting in poor post-meal regulation.

\begin{figure*}[t]
    \centering
    \includegraphics[width=\linewidth]{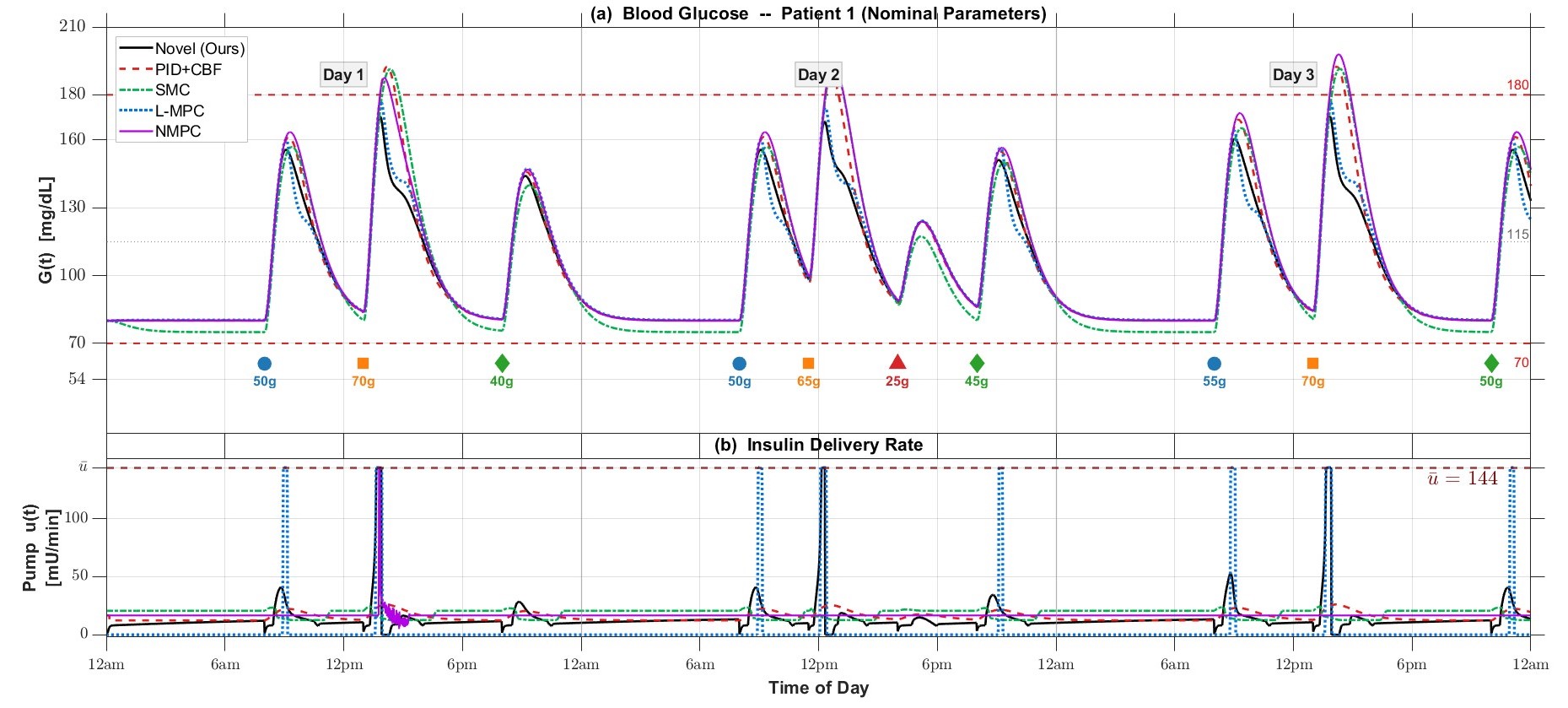}
    \caption{Closed-loop blood glucose $G(t)$ and insulin delivery rate $u(t)$ 
    for Patient~1 (nominal Bergman parameters) over the 72-hour, 10-meal 
    protocol under five controllers: Novel~(Ours), PID+CBF, SMC\, L-MPC, and 
    NMPC. Panel~(a) shows glucose trajectories with safety bounds (red dashed lines); panel~(b) shows the corresponding insulin pump delivery rate with the constraint $\overline{u} = 144$~mU/min.}
    \label{fig:benchmark}
\end{figure*}
\begin{figure*}[t]
    \centering
    \includegraphics[width=\linewidth]{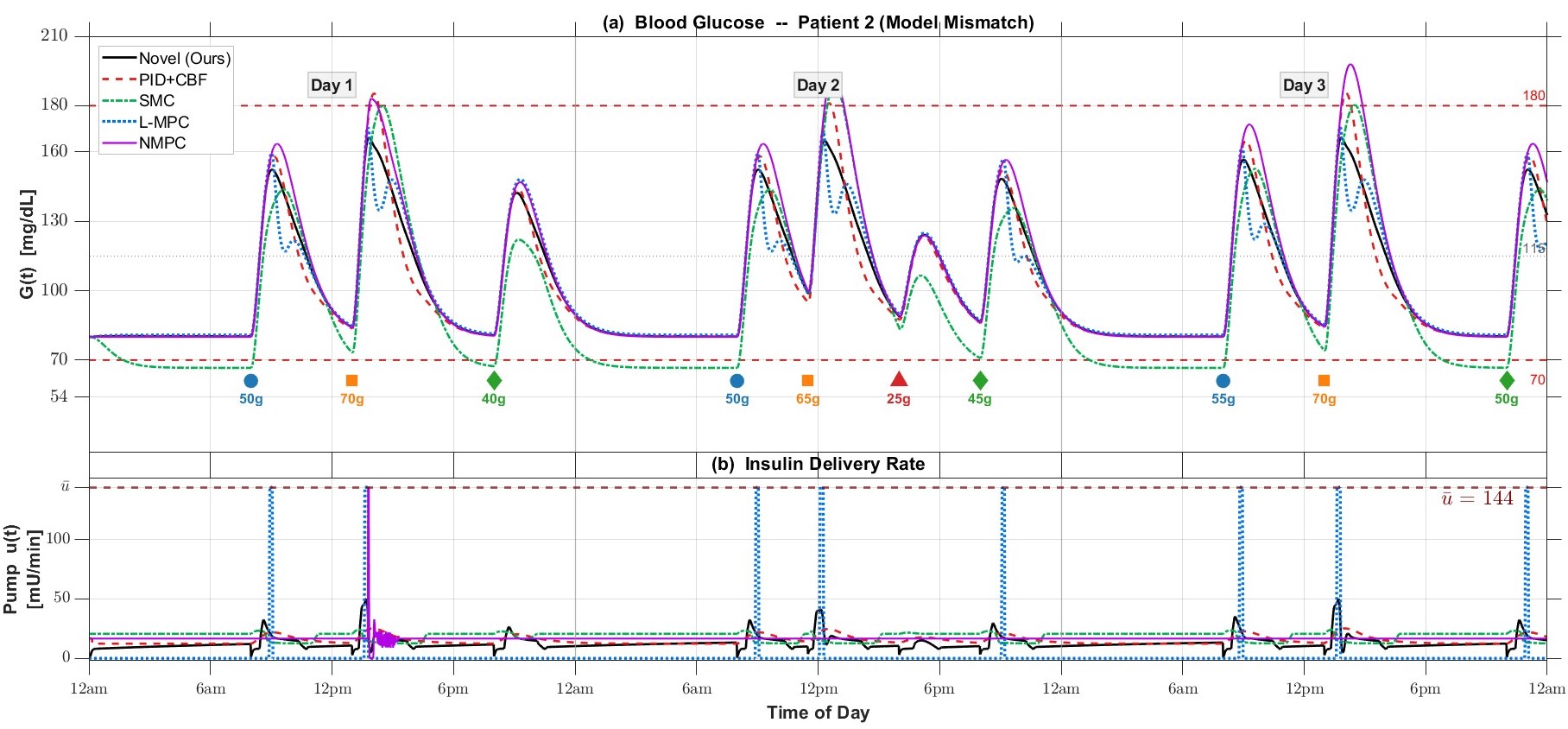}
    \caption{Robustness evaluation of the same five controllers as Figure~\ref{fig:benchmark} under model mismatch (Patient~2), while all controllers internally retain the nominal values. The proposed GSTC controller requires no plant parameters beyond $G_b$, and thus incurs no mismatch 
    penalty. Meals are unannounced to all controllers.}
    \label{fig:worst_benchmark}
\end{figure*}

\begin{table*}[t]
\centering
\footnotesize
\setlength{\tabcolsep}{2.0pt}
\renewcommand{\arraystretch}{1.0}
\caption{Closed-loop glycemic and computational performance. Protocol: 3-day, 10-meal (Day~1: regular; Day~2: busy with snack; Day~3: late dinner); meals unannounced ($d_{\mathrm{ctrl}}=0$ in all predictions). Patient~1: nominal Bergman model. Patient~2: insulin sensitivity $P_3\times3$; all baselines retain nominal $P_3$ (honest model mismatch). \textbf{Bold}: best value per column.}
\label{tab:benchmark}
\resizebox{0.9\textwidth}{!}{%
\begin{tabular}{l|rrrrr|rrrrr}
\toprule
 & \multicolumn{5}{c|}{\textbf{Patient~1 --- Nominal}} & \multicolumn{5}{c}{\textbf{Patient~2 --- Model Mismatch}} \\
\cmidrule(lr){2-6}\cmidrule(lr){7-11}
\textbf{Metric} & \textbf{GSTC (Ours)} & PID+CBF & SMC & L-MPC & NMPC & \textbf{GSTC (Ours)} & PID+CBF & SMC & L-MPC & NMPC \\
\midrule
\multicolumn{11}{l}{\textit{\textbf{A.\; Clinical Metrics}}} \\
TIR $[70,180]$ [\%] & \textbf{100.00} & 96.92 & 96.32 & 100.00 & 96.46 & \textbf{100.00} & 98.50 & 64.08 & 100.00 & 96.67 \\
TITR $[70,140]$ [\%] & \textbf{86.79} & 82.30 & 82.06 & 86.32 & 80.07 & 85.72 & 85.79 & 52.14 & \textbf{88.71} & 80.05 \\
TAR$_{180}$ $(180,250]$ [\%] & \textbf{0.00} & 3.08 & 3.68 & 0.00 & 3.54 & \textbf{0.00} & 1.50 & 1.27 & 0.00 & 3.33 \\
TAR$_{250}$ $>\!250$ [\%] & \textbf{0.00} & 0.00 & 0.00 & 0.00 & 0.00 & \textbf{0.00} & 0.00 & 0.00 & 0.00 & 0.00 \\
TBR$_{70}$ $[54,70)$ [\%] & \textbf{0.00} & 0.00 & 0.00 & 0.00 & 0.00 & \textbf{0.00} & 0.00 & 34.64 & 0.00 & 0.00 \\
\rowcolor{gray!12}\textbf{TBR$_{54}$ $<\!54$ [\%]} & \textbf{0.00} & 0.00 & 0.00 & 0.00 & 0.00 & \textbf{0.00} & 0.00 & 0.00 & 0.00 & 0.00 \\
\midrule
\multicolumn{11}{l}{\textit{\textbf{B.\; Statistical Metrics}}} \\
MAX [mg/dL] & \textbf{170.48} & 192.53 & 192.41 & 177.09 & 197.89 & \textbf{166.34} & 185.37 & 189.44 & 170.83 & 197.89 \\
IQR [mg/dL] & \textbf{42.44} & 43.03 & 50.91 & 42.49 & 48.32 & 42.89 & \textbf{36.57} & 52.64 & 39.35 & 48.32 \\
MIN [mg/dL] & \textbf{80.00} & 80.00 & 74.99 & 80.00 & 80.00 & \textbf{80.00} & 80.00 & 66.64 & 80.00 & 80.00 \\
Mean~G [mg/dL] & \textbf{102.97} & 105.21 & 103.79 & 103.27 & 107.15 & 103.36 & 102.46 & \textbf{96.40} & 102.64 & 107.12 \\
CV [\%] & \textbf{25.01} & 29.20 & 31.84 & 25.24 & 30.14 & 25.12 & 27.02 & 34.10 & \textbf{23.62} & 30.06 \\
Max~$u$ [mU/min] & 144.00 & 26.07 & \textbf{24.33} & 144.00 & 144.00 & 50.19 & 25.22 & \textbf{24.09} & 144.00 & 144.00 \\
\midrule
Comp. Time [ms/step] & 0.00225 & 0.00498 & 0.00209 & 1.04634 & 1.41472 & 0.00122 & 0.00397 & 0.00154 & 1.03662 & 1.41738 \\
\bottomrule
\end{tabular}}
\end{table*}

To evaluate robustness, the same controllers are tested on Patient~2 (Figure~\ref{fig:worst_benchmark}), whose insulin sensitivity is three times higher than the nominal model, while all the controllers continue to use the nominal parameters. GSTC again achieves a TIR of 100.00\% without any retuning. In contrast, the performance of the baseline controllers degrades significantly. SMC drops to a TIR of 64.08\%, with substantial time spent below range (TBR$_{70} = 34.64\%$) and a minimum glucose of 66.64~mg/dL. PID+CBF, despite being computationally faster, and NMPC both exceed the upper bounds post-lunch on all days, whereas L-MPC achieves 100.00\% TIR, although with higher computational and insulin delivery. 

Across both patients, GSTC consistently achieves the lowest peak post-meal glucose, the lowest CV, and the lowest eA1c among all controllers. It also delivers insulin in a smooth and bounded manner, using high input levels only when required by large meal disturbances, while maintaining minimal computational cost for easy deployment on wearable devices.

\begin{figure*}[t]
    \centering
    \includegraphics[width=\linewidth]{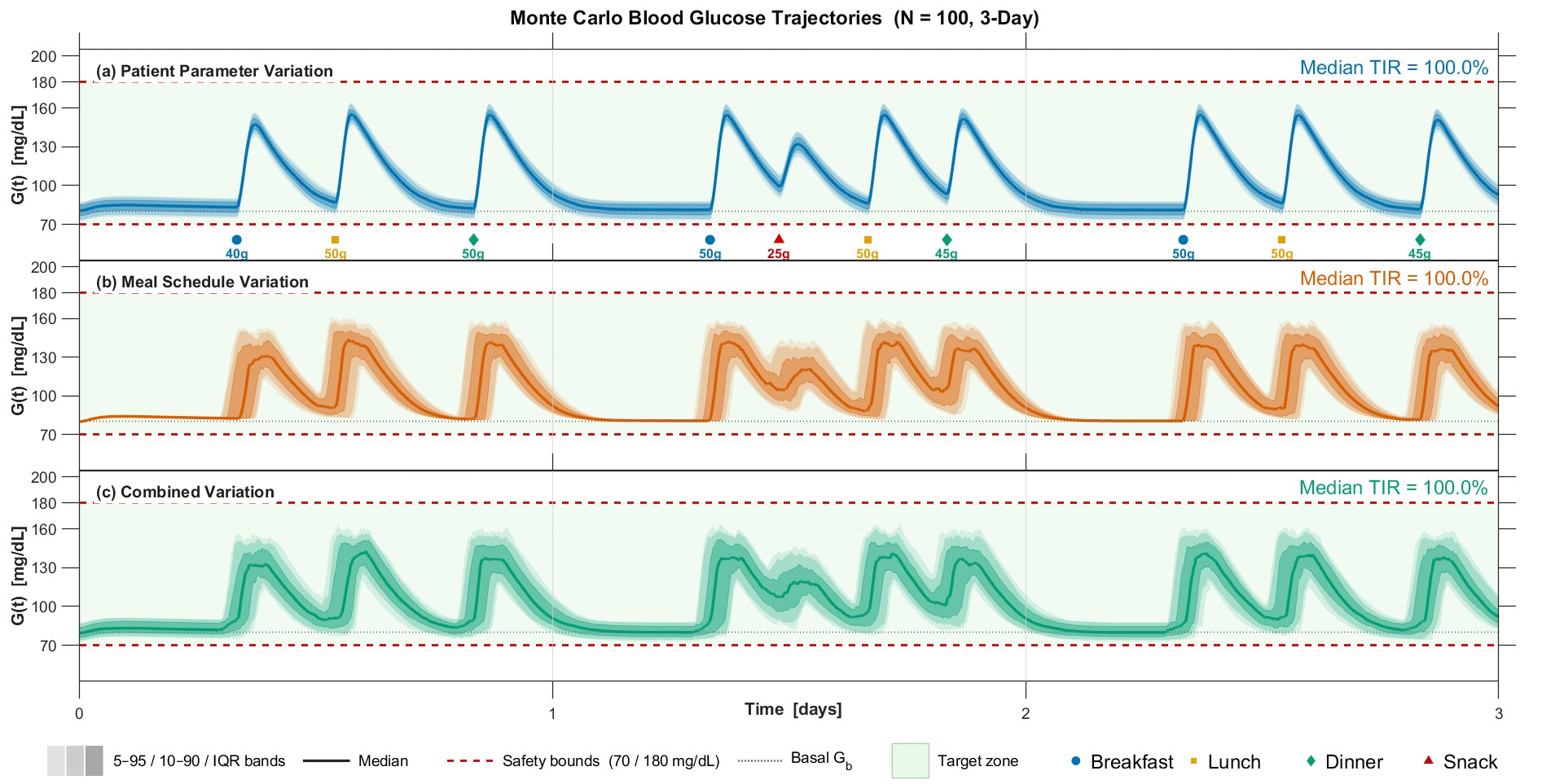}
    \caption{Monte Carlo blood glucose trajectories ($N = 100$, three-day horizon) 
    under three uncertainty scenarios: (a) patient parameter variation 
    ($\pm30\%$ in $S_G$, $P_2$, $P_3$; $\pm10\%$ in $G_b$, $I_b$), 
    (b) meal schedule uncertainty ($\pm60$ min timing, $\pm20$~g carbohydrate 
    content), and (c) combined parameter and meal uncertainty. Shaded bands 
    represent the 90th/80th percentile envelope and interquartile range (IQR) 
    of the ensemble; the solid line is the ensemble median. Safety bounds 
    $[70,\, 180]$~mg/dL are shown as red 
    dashed lines.}
    \label{fig:montecarlo}
\end{figure*}

\subsection{Robustness to Parameter and Meal Variability}

The previous comparison showed performance under two fixed parameters. We now examine whether these safety and performance properties hold when both patient physiology and meal patterns vary. To do this, we conduct a Monte Carlo study with 300 closed-loop simulations, divided into three uncertainty scenarios, with each scenario evaluated over a three-day horizon for $N = 100$ virtual patients.

Scenario~1 considers only physiological uncertainty by applying independent uniform variations of $\pm 30\%$ to the patient parameters $S_G$, $P_2$, and $P_3$, and $\pm 10\%$ to the basal values $G_b$ and $I_b$, consistent with reported variability in the Bergman model. Scenario~2 considers only meal uncertainty by varying meal timing by $\pm 60$~minutes and carbohydrate content by $\pm 15$~g per meal, while keeping patient parameters fixed. Scenario~3 combines both sources of uncertainty and represents the most realistic operating condition.

\begin{table}[t]
\centering
\caption{Monte Carlo Summary ($N = 100$, 3-day simulation)}
\label{tab:montecarlo}
\begin{tabular}{lccc}
\toprule
\textbf{Metric} & \textbf{Scenario 1} & \textbf{Scenario 2} & \textbf{Scenario 3} \\
\midrule
TIR [70--180] [\%]        & 100.00 $\pm$ 0.00  & 100.00 $\pm$ 0.00  & 100.00 $\pm$ 0.00  \\
TITR [70--140] [\%]       & 83.52 $\pm$ 4.73   & 85.03 $\pm$ 1.96   & 85.67 $\pm$ 5.17   \\
Mean G [mg/dL]            & 109.18 $\pm$ 4.98  & 107.93 $\pm$ 1.00  & 107.57 $\pm$ 5.25  \\
SD G [mg/dL]              & 24.82 $\pm$ 0.85   & 24.80 $\pm$ 0.72   & 24.58 $\pm$ 0.92   \\
CV [\%]                   & 22.77 $\pm$ 1.17   & 22.97 $\pm$ 0.51   & 22.90 $\pm$ 1.29   \\
Peak PP [mg/dL]           & 151.20 $\pm$ 5.37  & 150.27 $\pm$ 2.48  & 150.37 $\pm$ 6.21  \\
Max $u$ [mU/min]          & 117.22 $\pm$ 25.24 & 143.49 $\pm$ 3.44  & 140.05 $\pm$ 13.01 \\
\bottomrule
\end{tabular}%
\end{table}

The results across all 300 simulations are summarized in Table~\ref{tab:montecarlo}, with the full set of trajectories shown in Figure~\ref{fig:montecarlo}. 
In all scenarios, TIR remains $100\%$, and the coefficient of variation remains around $25\%$, well below the clinical threshold of $36\%$.
The responses remain tightly clustered around the median, narrowing during fasting periods and widening temporarily after meals before quickly recovering. As expected, Scenario~3 shows a slightly greater spread due to combined uncertainty, but all trajectories remain within the prescribed safety bounds. The input constraint is also satisfied in all simulations, with $u(t)$ always below $\overline{u} = 144$~mU/min.

These results show that the safety and performance properties established for the nominal case extend to a wide range of physiological and meal variability. Since the controller is designed using the three-state Bergman model, the next subsection examines whether this robustness carries over to a more complex glucose--insulin model.

\subsection{Preclinical Validation on High-fidelity Simulator}

The Monte Carlo results demonstrated reliability under uncertainty within the Bergman model. A more important question is how the proposed GSTC controller performs under a structurally different higher-order model. To test this, the GSTC controller is directly applied to a high-fidelity 13-state multi-compartmental simulator~\cite{xie2018simglucose, DallaMan2014}, which is the accepted \textit{in silico} standard for AP evaluation. This serves as a deliberate stress test, where controllers derived for a simple three-state model are applied to a much more complex system.

Two operating conditions are evaluated on six virtual patients under two conditions: a fasting protocol (no meals) to study steady-state behavior under model mismatch, and a three-meal protocol (40--60--40~g carbohydrate, totaling 140~g/day) to evaluate disturbance rejection. The results are summarized in Table~\ref{tab:clinical_sim}.
\begin{table*}[t]
\caption{Glycemic outcomes under fasting and post-prandial conditions on the High-Fidelity Simulator}
\label{tab:clinical_sim}
\centering
\begin{tabular}{l cccc cccc}
\toprule
\multirow{2}{*}{\textbf{Patient}} & \multicolumn{4}{c}{\textbf{0-meal}} & \multicolumn{4}{c}{\textbf{3-meal (40--60--40 g CHO)}} \\
\cmidrule(lr){2-5} \cmidrule(l){6-9}
 & \textbf{Mean $\pm$ SD} & \textbf{TIR\%} & \textbf{TBR$_{70}$\%} & \textbf{TAR$_{180}$\%} & \textbf{Mean $\pm$ SD} & \textbf{TIR\%} & \textbf{TBR$_{70}$\%} & \textbf{TAR$_{180}$\%} \\
\midrule
Patient 1 & 141.1 $\pm$ 14.5 & 100.00 & 0.00 & 0.00 & 132.8 $\pm$ 30.6 & 92.08 & 0.00 & 7.92 \\
Patient 2 & 138.6 $\pm$ 14.9 & 99.79 & 0.00 & 0.21 & 137.6 $\pm$ 30.3 & 91.88 & 0.00 & 8.13 \\
Patient 3 & 136.5 $\pm$ 21.7 & 98.96 & 0.00 & 1.04 & 149.0 $\pm$ 26.2 & 89.58 & 0.00 & 10.42 \\
Patient 4 & 145.4 $\pm$ 13.2 & 99.17 & 0.00 & 0.83 & 153.2 $\pm$ 21.7 & 90.83 & 0.00 & 9.17 \\
Patient 5 & 137.8 $\pm$ 18.2 & 100.00 & 0.00 & 0.00 & 143.2 $\pm$ 22.9 & 97.29 & 0.00 & 2.71 \\
Patient 6 & 135.3 $\pm$ 19.3 & 100.00 & 0.00 & 0.00 & 143.8 $\pm$ 20.7 & 98.75 & 0.00 & 1.25 \\
\midrule
Mean & 139.1 $\pm$ 17.0 & 99.65 & 0.00 & 0.35 & 143.3 $\pm$ 25.4 & 93.40 & 0.00 & 6.60 \\
\bottomrule
\end{tabular}%
\quad\\
\vspace{1ex} 
{\raggedright \small \textit{Note:} TIR = time in range (70--180 mg/dL); TBR$_{70}$ = time below 70 mg/dL; TAR$_{180}$ = time above 180 mg/dL; SD = standard deviation; CHO = carbohydrate. TBR$_{54}$ and TAR$_{250}$ = 0.00\% in all cases (not shown).\par}
\end{table*}

Under fasting conditions, the controller achieves a mean TIR of 99.65\%, far exceeding the clinical target of $70\%$~\cite{Battelino2019}. No hypoglycemic events occur in any of the six patients, with TBR is $0\%$ throughout. Under the three-meal protocol, the mean TIR across patients is 93.40\%, well above the clinical target of 70\%~\cite{Battelino2019}. Importantly, severe hypoglycemia (TBR$_{54}$) and severe hyperglycemia (TAR$_{250}$) are zero for all patients under both protocols. The TAR values in the post-meal case, ranging from 2.71\% to 10.42\%, indicate that deviations after meals, although not fully eliminated, are still contained within a non-severe range. This is expected due to model mismatch and the absence of meal anticipation.

These results provide empirical evidence that the GSTC controller performs well and retains its safety and performance properties when directly applied to more realistic, higher-order models.

\begin{remark}
The formal guarantees of Theorem~\ref{thm:glucose_control} are derived for the three-state Bergman model. Extending these to the full 13-state glucose--insulin model remains an open problem and is part of future work. The results presented here are an intentional model-mismatch stress test to empirically assess robustness beyond the design model.
\end{remark}

\section{Discussion and Future Work}
\label{sec:discussion}

The results above consistently highlight three key properties that distinguish proposed GSTC from existing approaches.

Across all scenarios, GSTC consistently maintains glucose within clinically established safe bounds. This performance is not the result of tuning, but follows directly from the formal guarantees of the control design and the feasibility conditions. In contrast, existing approaches do not provide such guarantees, which is reflected in their performance under disturbances and model mismatch.

Due to its model-free nature, GSTC does not rely on patient-specific parameters and remains robust to both patient variability and unannounced meal disturbances. In comparison, model-based approaches degrade significantly under model mismatch, which is unavoidable in practice.

The closed-form structure of GSTC also leads to very low computational cost. The controller runs several orders of magnitude faster than optimization- and learning-based methods, making it well-suited for real-time implementation on low-cost wearable devices.

The preclinical validation on the 13-state simulator further supports these findings. Although the controller is designed using the three-state Bergman model, it is applied without retuning to a more complex system. The results show a mean TIR of 99.65\% under fasting conditions and 93.40\% under three unannounced meals, both above the clinical target of 70\%~\cite{Battelino2019}, with no hypoglycemic events.

Building on these results, several directions can further advance the framework toward clinical translation. First, the formal safety guarantees can be extended from the three-state Bergman model to higher-dimensional and more realistic glucose--insulin models. Second, incorporating subcutaneous insulin delivery dynamics and CGM measurement delays is important, as these introduce an effective lag between actuation and physiological response. Integrating predictive elements to anticipate glucose excursions and take proactive control actions, while retaining formal safety guarantees, is a key next step. Finally, the transformation function used in this work represents one instance within a broader class. Future work will explore learning more control-efficient transformation functions, for example, using physics-informed neural networks (PINNs). Moving toward hardware-in-the-loop validation is a natural next step toward clinical deployment.

\section{Conclusion}
\label{sec:conclusion}

This paper presented the Glycemic Safety Tube Control (GSTC) framework, a closed-form and model-free approach for automated insulin delivery that guarantees glucose remains within prescribed physiological bounds under bounded meal disturbances and parameter uncertainty. The cascaded control design enforces safety across the glucose, remote insulin, and plasma insulin subsystems without requiring patient-specific identification, online optimization, or iterative computation. The derived feasibility conditions provide a clear link between controller parameters, disturbance bounds, and the resulting safety guarantees.

The evaluation shows that GSTC maintains perfect glucose containment across both nominal and mismatched patients, as well as across all Monte Carlo uncertainty scenarios, with zero hypoglycemic events. When applied to a structurally different high-fidelity simulator, the controller continues to perform well, achieving TIR values well above clinical targets while maintaining safety.

\bibliographystyle{unsrt} 
\bibliography{cas-refs} 

\appendix
\section{Appendix}
\begin{figure*}[htbp]
    \centering
    \includegraphics[width=0.8\linewidth]{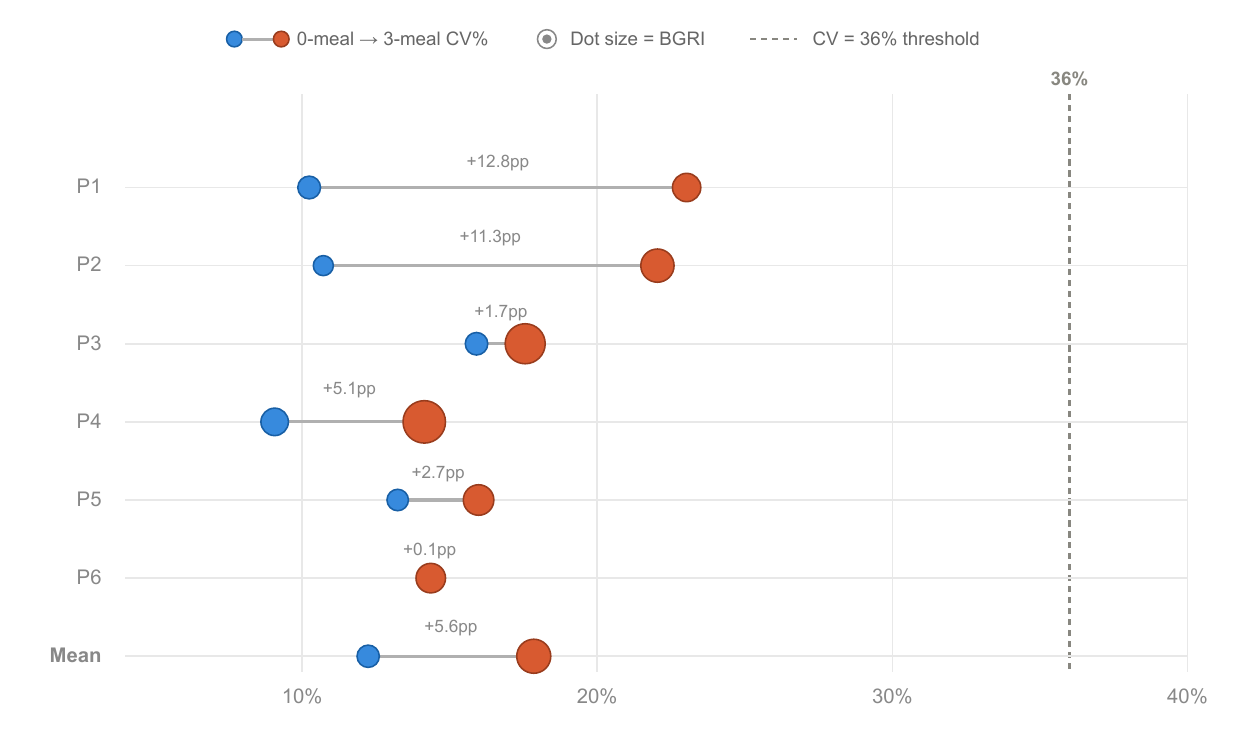}
    \caption{Intra-patient shift in coefficient of variation
    (CV\%) from fasting (0-meal, blue) to post-prandial (3-meal,
    coral) conditions~\cite{Battelino2019}. Line length encodes
    the magnitude of meal-induced variability increase. The dashed
    vertical line marks the 36\% instability threshold; all patients
    remain below this boundary under both conditions. Patients~1
    and~2 show the largest meal-induced escalation (+12.8 and +11.3
    percentage points respectively), indicating pronounced
    postprandial sensitivity, while Patients~3 through~6 exhibit
    more contained responses, with Patient~6 showing negligible
    meal-induced variability (+0.1 percentage points).}
    \label{fig:clinical_sim}
\end{figure*}
Glycemic variability remained well-controlled throughout the pre-clinical validatin on high fidelity Simulator. As
shown in Figure~\ref{fig:clinical_sim}, all six patients remained
below the 36\% CV instability threshold~\cite{Battelino2019}
under both fasting and post-prandial conditions. The mean CV
increase from fasting to three-meal conditions was 5.6
percentage points. Patients~1 and~2 exhibited the largest
meal-induced CV escalation (+12.8 and +11.3 percentage points
respectively), reflecting pronounced postprandial sensitivity
in those profiles, while Patients~3 through~6 showed more
modest increases, with Patient~6 exhibiting near-zero
meal-induced variability (+0.1 percentage points).

\end{document}